\documentclass{article}

\usepackage{PRIMEarxiv}
\usepackage{amsthm,amsmath}
\usepackage{amssymb}
\usepackage{amscd}
\usepackage{latexsym}
\usepackage{csquotes}
\usepackage{mathtools}
\usepackage[utf8]{inputenc} 
\usepackage[T1]{fontenc}    
\usepackage{hyperref}       
\usepackage{url}            
\usepackage{booktabs}       
\usepackage{amsfonts}       
\usepackage{nicefrac}       
\usepackage{microtype}      
\usepackage{lipsum}
\usepackage{fancyhdr}       
\usepackage{graphicx}       
\graphicspath{{media/}}     
\newtheorem{thm}{\bf Theorem}[section]
\newtheorem{corollary}[thm]{\bf Corollary}
\newtheorem{theorem}[thm]{\bf Theorem}
\newtheorem{lemma}[thm]{\bf Lemma}
\newtheorem{proposition}[thm]{\bf Proposition}
\newtheorem{definition}[thm]{\bf Definition}
\newtheorem{remark}[thm]{\bf Remark}
\newtheorem{example}[thm]{\bf Example}

\title{Majorization  and Semi-Doubly Stochastic Operators on $L^1(X)$
}

\author{
  Seyed Mahmoud Manjegani, Shirin Moein \\
   Department of Mathematical Sciences,\\ Isfahan University of Technology, \\
  Isfahan, Iran\\
  \texttt{ manjgani@cc.iut.ac.ir} \\
  \texttt{ s.moein@math.iut.ac.ir} 
}

\begin{document}
\maketitle

\begin{abstract}
This article is devoted to a study of majorization based on semi-doubly stochastic operators (denoted by $S\mathcal{D}(L^1)$) on $L^1(X)$ when $X$ is a $\sigma$-finite measure space. We answered Mirsky's question and characterized the majorization by means of semi-doubly stochastic maps on $L^1(X)$. We collect some results of semi-doubly stochastic operators such as a strong relation of semi-doubly stochastic operators and integral stochastic operators, and relatively weakly compactness of $S_f=\{Sf: ~S\in S\mathcal{D}(L^1)\}$ when $f$ is a fixed element in $L^1(X)$ by proving equi-integrability of $S_f$. 
\end{abstract}


\section{Introduction}

 Until recent decades, the main attention in majorization theory was paid to finite-dimensional space, but recently because of its significant applications in a broad spectrum of fields, especially in quantum physics, considerable interest in infinite-dimensional spaces appeared mathematically and physically \cite{Semi, KAFTAL2, jadid, MM,  Rajesh}.\\
 The aim of this paper is to study the notion of majorization on $L^1(X,\mu)$, that is the space of all absolutely integrable function $f:X\to \mathbb{R}$ when $(X,\mu)$ is $\sigma$-finite measure space. 
Our motivation to work on this space is its application in quantum information theory, for details you can see \cite{MM, epsilon}. We start with short history.

  Hardy, Littlewood, and P{\'o}lya extended an equivalent condition of Muirhead's inequality from non-negative integer vectors to real vectors and called it vector majorization as follows. \\
Let $X,Y\in\mathbb{R}^n$ with the similar total of the whole components. $X$ is vector majorized by $Y$ (denoted by $X\prec Y$) if for each $k\in\{1,2,\ldots, n\}$ the sum of the $k$ largest components of $X$ is less than or equal to the sum of the $k$ largest components of $Y$. 

To avoid difficulty of decreasing rearrangement of components Hardy, Littlewood and P{\'o}lya proved  equivalent conditions for $X\prec Y$ independent of decreasing rearrangements as following theorem.
\begin{theorem}\label{majvect}\cite{Marshal}
Let $x,Y\in \mathbb{R}^n$. Then the following statements are equivalent.
\begin{itemize}
    \item[(1)] $X\prec Y$.
      \item[(2)] There exists a doubly stochastic matrix $D=[d_{ij}]$ (an $n$-square matrix with $d_{ij}>0$, $\sum_{i=1}^n d_{ij}=1$ and $\sum_{j=1}^n d_{ij}=1$ for all $i,j\in \{1,2,\ldots, n\}$) such that $X=DY$.
        \item[(3)] The following inequality hold for all convex functions $g$,
        $$\sum_{i=1}^n g(x_i)\leq \sum_{i=1}^n g(y_i).$$ 
        Here $x_i$ and $y_i$ for all $i=1,\cdots, n$ are components of $X$ and $Y$ respectively.
\end{itemize}
\end{theorem}
As an infinite counterpart of Hardy, Littlewood, and P{\'o}lya's results, Mirsky proposed this question\cite[Section 4: Page 328]{Mirsky}:
\begin{displayquote}
\textit{``The introduction of infinite doubly stochastic matrices raises the question whether there 
exists an infinite analogue 
of Hardy, Littlewood, and P{\'o}lya's results?''}    
\end{displayquote}

In this work, by using semi-doubly stochastic operators we completely answer Mirsky's question (extension of Hardy, Littlewood, and P{\'o}lya's results, Theorem \ref{majvect}) on $L^1(X,\mu)$ when $(X,\mu)$ is $\sigma$-finite measure space. 
We denote the dual space of $L^1 (X,\mu)$ which is the space of essentially bounded functions on $X$ by $L^\infty (X,\mu)$ and to shorten notation, we write $L^1(X)$ and $L^\infty(X)$ when no confusion can arise.

Section \ref{Sec2} contains a brief historical summary of results and mathematical preliminaries on majorization on $L^1(X)$. In Section \ref{Section2} we will look more closely at semi-doubly stochastic operators and provide a method of constructing them in Theorem \ref{SDSprop*} 
and also for fixed $f\in L^1(X)$ we will prove equi-integrability of $S_f=\{Sf:~~S\in S\mathcal{D}(L^1(X))\}$ in Theorem \ref{equiint} for arbitrary measure space $X$, and using this theorem leads to Corollary \ref{Wcompact} that is relatively weakly compactness of $S_f$ when $X$ is finite measure space. 
In Section \ref{Sec4} some of the recent results like relation between majorization and integral operators are
reviewed in a more general setting by using semi-doubly stochastic operators and we we will answer to Mirsky's question by giving a full characterization of majorization in Theorem \ref{main2}.

\section{Mathematical preliminaries}\label{Sec2}
There are two main ways for generalizing the majorization concept on measurable functions, extension based on decreasing rearrangement and based on stochastic operators. 
\subsection{Decreasing Rearrangement} \label{1}

In 1963, definition of decreasing rearrangement for measurable functions on $([0,1],m)$ where $m$ is the Lebesgue measure introduced by Ryff \cite{Ryff} and then in 1974 extended by Chong \cite{Chong} to $L^1(X,\mu)$ for an arbitrary measure space $(X,\mu)$ as follows.  
\begin{definition}
If $f$ is any measurable (respectively non-negative integrable) \linebreak function defined on a finite (respectively infinite) measure space $( X, \mu)$, then there exists a unique right continuous decreasing function $f^\downarrow$ on the interval $[0, \mu(X)]$, called the decreasing rearrangement of $f$ and defined by 
\begin{eqnarray*}
f^\downarrow (s)&=& \inf \{t: d_f(t)\leq  s\},\quad \quad\quad 0 \leq  s \leq \mu(X)\\
&=& \sup \{t: d_f(t)> s\}, \quad \quad\quad 0 \leq  s \leq \mu(X)
\end{eqnarray*}
where  $d_f$ is called distribution function of $f$ and for all real $t$ defined by 
\[d_f(t)= \mu\{x: f(x)>t\}.\]
\end{definition}
Then Chong generalized the notion of majorization as following definition to \linebreak $L^1(X,\mu)$ for an arbitrary measure space $X$. 
\begin{definition}\cite{Chong}\label{DecRe} Let $(X, \mu)$ be an arbitrary measure space and $f,g\in L^1(X)$ (notice that for the infinite measure space we have to suppose $f,g$ are non-negative). Then we say that $f$ is weak majorized by $g$ and write $f\prec_w g$ if 
\[
\int_0^s f^\downarrow dm \leq  \int_0^s g^\downarrow dm,\quad \forall\, 0\leq  s\leq  \mu(X)\]
and furthermore if we have the following equality we say that $f$ is majorized by $g$ and write $f\prec g$.
\[\int_0^{\mu(X)} f^\downarrow dm = \int_0^{\mu(X)} f^\downarrow dm, 
\] where $dm$ is the Lebesgue measure on the interval $[0,\infty)$.
\end{definition}
The following Theorems are due to Chong which we will use them in next section.
\begin{theorem}\label{Chong2}\cite[Theorem 1.6]{Chong} If $f,g\in L^1(X,\mu)$ are non-negative and $\mu(X)$  is infinite, then
 for all $u\in \mathbb{R}$
$$
\int_{0}^{t}f^{\downarrow}dm\leq\int_{0}^{t}g^{\downarrow}dm,\ \forall t\in[0,\infty),
$$
 if and only if
$$
\int_{u}^{\infty}d_{f}dm\leq\int_{u}^{\infty}d_{g}dm.
$$
\end{theorem}

\begin{theorem}\cite[Corollary 1.2]{Chong}\label{Chong}
  If $f$ is a measurable function on $(X,\mu)$, where $\mu(X)$ is finite or infinite, then
$$
\int_{X}\max\{(f-u),0\}d\mu=\int_{u}^{\infty}d_{f}(t)dt.
$$
\end{theorem}
\subsection{Semi-Doubly Stochastic operators}\label{2}
 To avoid the difficulty of decreasing rearrangement in the definition of majorization, Ryff \cite{Ryff} introduced and characterized an important class of linear operators $T:L^1([0,1],m)\to L^1([0,1],m)$ which is known as doubly stochastic operators such that $Tf$ is majorized by $f$ for all $f\in L^1([0,1],m)$.
 Ryff's characterization can not be extended on $L^1(X, \mu)$ when $(X, \mu)$ is $\sigma$-finite measure space(see \cite[Example II.7]{MM} as a counter-example), but it could be extended by using semi doubly stochastic operators which is a new class of operators that introduced by Manjegani and et.al \cite{MM} on $l^1$ space and extended  by Bahrami and et.al \cite{Semi} on $L^1(X, \mu)$ where $(X, \mu)$ is $\sigma$-finite measure space.
 
  The class of semi-doubly stochastic operators is larger than the class of doubly stochastic operators and smaller than the class of integral preserving operators which is known as Markov operators or stochastic operators in references. It is worth noting that the theory of Markov operators is extremely rich, one of its applications is an examination of the eventual behavior of densities in dynamical systems for more details see \cite{Markov}. 
  
  The definition of doubly stochastic, semi-doubly stochastic, and Markov operators are as follows, but before that, we recall that each bounded linear map $T: ~ X\to Y,$ between two normed linear spaces $X$ and $Y$, induces a bounded linear operator $T^*:Y^*\to X^*$, between the dual spaces, defined for all $g\in Y^*$ and $x\in X$ by $\langle x, T^*g\rangle=\langle Tx, g\rangle$, which $\langle \cdot ,\cdot \rangle$ denotes the dual pairing between the dual spaces. 
  
 \begin{definition}\label{Semi} 
Let $(X,\mu)$ be an arbitrary measure space. 
 \begin{itemize}
 \item[(a)]\label{defMarkov}  A positive operator $T:L^1(X)\to L^1(X)$ is Markov operator if it satisfies for all $f\in L^1(X)$,
$$\int_X Tf ~d\mu=\int_X f ~d\mu.$$
In ordinary terms $T^*(1)=1$.
The set of all Markov operators on $L^1(X)$ denoted by $\mathcal{M}(L^1(X))$.
 \item[(b)] \label{defsemi} $T\in \mathcal{M}(L^1(X))$ is semi-doubly stochastic operator if it satisfies for each $E\in \mathcal{A}$ with $\mu(E)<\infty$, \[\int_X T^*\chi_E ~d\mu\leq \mu(E).\]
 The set of all semi-doubly stochastic operators on $L^1(X)$ denoted by $S\mathcal{D}(L^1(X))$.
 \item[(c)] \label{defdoubly} $T\in S\mathcal{D}(L^1(X))$ is doubly stochastic operator if it satisfies for each $g\in L^\infty(X)$, \[\int_X T^*g ~d\mu= \int_X g ~d\mu.\]
 The set of all doubly stochastic operators on $L^1(X)$ denoted by $\mathcal{D}(L^1(X))$.
 \end{itemize} 
\end{definition}
In the following lemma, we easily prove that for arbitrary either finite or infinite measure space $X$ the Markov operators are bounded, and therefore all semi-doubly stochastic and doubly stochastic operators are bounded.
\begin{lemma}
Let $T\in \mathcal{M}(L^1(X))$, then $||T||=1$, where 
$$||T||=\sup \{||Tf||_1:~~||f||_1=1\}.$$
\end{lemma}
\begin{proof}
By definition of Markov operator for each $f\in L^1(X)$,
$$\int_X Tf ~d\mu=\int_X f ~d\mu.$$ 
Therefore 
\begin{align*}
||T||&= \sup \{||Tf||_1:~~||f||_1=1\}
= \sup \{\int_X Tf ~d\mu:~~||f||_1=1\}\\&=\sup \{\int_X f ~d\mu:~~||f||_1=1\} = \sup \{||f||_1:~~||f||_1=1\}=1 \end{align*} 
\end{proof}
By adopting counting measure $\mu$ on the space of natural numbers $\mathbb{N}$, we rewrite the previous definition for the Banach space $l^1$, consisting of all sequences whose series are absolutely convergent. In fact majorization theory on $l^1$ space play a key role in quantum information theory, for instant the generalization of Nielsen’s result for infinite dimensional quantum system asserts that quantum states $\phi$ can be convertible to quantum states $\psi$ if and only if  the sequence of Schmidt coefficients of $\phi$ is majorized by the sequence of Schmidt coefficients of $\psi$ \cite[Theorem 1]{Rajesh}, and by infinite dimensional version of Schmidt decomposition theorem the sequences of Schmidt coefficients are belong to $l^1$ space. Therefore rewriting of above definition would be useful to avoid confusion for readers in quantum information theory area.
Before that, we denote sequence $e_n\in l^1$ as well known Kronecker delta $\delta_{mn}=(e_n)_m$ for each $m,n\in \mathbb{N}$, and also we denote the dual pairing between $l^1$ and its dual space $l^\infty$
with $\langle\cdot,\cdot\rangle: l^1\times l^\infty\rightarrow \mathbb{R}$ which for each $f\in l^1$ and $g\in l^\infty$, $\langle f,g\rangle=\sum_{i\in \mathbb{N}}f_ig_i$. If we consider $e_i$ as an element of $l^\infty$
then for each $f\in l^1$, $\sum_{i\in \mathbb{N}} \langle f, e_i\rangle=\sum_{i\in \mathbb{N}} f_i$.
We are now ready to rewrite above definition on $l^1$ space. 
\begin{definition}\label{Quantum def}
 A positive operator $D: l^1 \rightarrow l^1$ is 
\begin{itemize}
    \item[(a)] 
A Markov operator if
$\sum_{i=1}^\infty \langle De_j, e_i\rangle=1$.
  \item[(b)] 
A semi-doubly stochastic operator if 
$\sum_{i=1}^\infty \langle De_j, e_i\rangle=1$,
and 
$\sum_{j=1}^\infty \langle De_j, e_i\rangle\leq 1$.
  \item[(c)] 
A doubly stochastic operator if 
$\sum_{i=1}^\infty \langle De_j, e_i\rangle=1$,
and 
$\sum_{j=1}^\infty \langle De_j, e_i\rangle=1$.
\end{itemize}
\end{definition}
In general form it is obvious that 
\begin{equation}\label{incl}
    \mathcal{D}(L^1(X))\subseteq S\mathcal{D}(L^1(X))\subseteq \mathcal{M}(L^1(X)).\end{equation}
 It is worth referring to an important result by Bahrami et.al \cite[Proposition 2.6]{Semi} that if $\mu(X)<\infty$ then on $L^1(X, \mu)$ the semi-doubly stochastic operator coincide with the doubly stochastic operator. We provide a counterexample that asserts that in general the converse of inclusions \eqref{incl} is not true.
\begin{example}\label{Example1}
Let $\mu$ be counting measure on $X=\mathbb{N}$. Then positive operators $T_1,T_2,T_3:l^1\to l^1$ for each sequence $(a_n)_{n\in\mathbb{N}}\in l^1$ defined by 
$$T_1(a_n)=(\sum_{i=1}^\infty a_n,0,0,\dots),$$
$$T_2(a_n)=(0,a_1,a_2,\dots),$$
$$T_3(a_n)=(a_1,a_2,a_3,\dots).$$
Easily can be seen that 
$T_1\in\mathcal{M}(l^1), T_2\in S\mathcal{D}(l^1), T_3\in\mathcal{D}(l^1)$.
$T^*_1:l^\infty\to l^\infty$ defined by $T^*_1((b_n))=(b_1,b_1,\dots)$ and therefore
$T^*_1(\chi_{\{1\}})=T^*_1(e_1)=(1,1,\dots)$. Hence 
$T_1\notin S\mathcal{D}(l^1)$.
and also $T^*_2:l^\infty\to l^\infty$ as left shift operator define as $T^*_2((b_n))=(b_2,b_3,\dots)$. Hence 
$T^*_2(e_1)=(0,0,\dots)$ and $\int T^*_2(e_1)~d\mu=0<\int_{\mathbb{N}}e_1 ~d\mu=1$ therefore $T_2\notin \mathcal{D}(l^1)$.
\end{example}
Now the following theorem is to bring together two areas of Subsections \ref{1} and \ref{2}, that is an equivalent condition for majorization concept based on decreasing rearrangement (Definition \ref{DecRe}) by using semi-doubly stochastic operators (Definition \ref{Semi}).
 \begin{theorem}\cite[Corollary 2.10]{Semi}\label{main} Let $X$ be a $\sigma$-finite measure space and $f,g$ be non-negative belongs to $L^1(X)$. Then
$g \prec f$ if and only if there is a sequence $(S_n)_{n\in\mathbb{N}}$ in $S\mathcal{D}(
L^1(X))$ such that $S_nf \to g$ in $L^1(X)$.
\end{theorem}

\section{ Some Results on Semi-Doubly Stochastic Operators }\label{Section2}
\subsection{Method of Constructing Semi-Doubly Stochastic Operators} In the following theorem, we present a method for constructing doubly (and semi-doubly) stochastic operators. To this end, we assume that $(X,\mathcal{A},\mu)$ is a $\sigma$-finite measure space. Also we assume that $\{A_n;\quad n\in \mathbb{N}\}$ is a sequence of measurable sets with finite measure such that $X=\cup_{n=1}^\infty A_n$ and for each $n\in \mathbb{N}$, $A_n\subseteq A_{n+1}$. We denote the set of all measurable simple functions in $L^1(X)$ by  $\mathcal{S}$.

\begin{theorem}\label{prop*}
Let $D:\mathcal{S}\to L^1(X)$ be a linear function. Then $D$ has a unique extension to a doubly stochastic operator if and only if $D$ is nonnegative and the following inequalities hold for each $E\in \mathcal{A}$ with $\mu(E)<\infty$:
\begin{equation}\label{E1}
\int_X D\chi_E d\mu=\mu(E),
\end{equation}
\begin{equation}\label{E2}
\lim_{n\to\infty}\int_X \chi_E D\chi_{A_n}d\mu=\mu(E).
\end{equation}
\end{theorem}
\begin{proof}
Let $\varphi=\sum_{i=1}^n a_i \chi_{E_i}$ be a simple function. Then according to \eqref{E1}, 
\begin{equation}\label{E3}
\int_X D\varphi d\mu=\int_X D(\sum_{i=1}^n a_i\chi_{E_i})d\mu=\sum_{i=1}^n a_i\int_X D\chi_{E_i}d\mu=\sum_{i=1}^n a_i\mu(E_i)=\int_X \varphi d\mu,
\end{equation}
therefore for each $\varphi\in S$,
\[\|D \varphi\|_1=\int_X \vert D\varphi\vert d\mu\leq\sum_{i=1}^n \vert a_i\vert\int_X D\chi_{E_i} d\mu=\sum_{i=1}^n \vert a_i\vert\mu(E_i)=\int_X \vert \varphi\vert d\mu=\|\varphi\|_1.\]
Hence operator $D:\mathcal{S}\to L^1(X)$ is bounded and $\|D\|\leqslant 1$. Since $\mathcal{S}$ is dense in $L^1(X)$, $D$ has a unique extension to $L^1(X)$ which we denote it again by $D$. 

Due to \eqref{E3}, for each $f\in L^1(X)$,
\[\int_X Df d\mu=\int_X f d\mu,\]
and by using \eqref{E2}, for each $\varphi\in \mathcal{S}\subset L^1(X)\cap L^{\infty}(X)$,
\begin{align}\label{t}
\int_X D^*\varphi d\mu &= \lim_{n\to\infty}\int_X D^*\varphi \chi_{A_n}d\mu = \lim_{n\to\infty}\int_X\varphi D\chi_{A_n}d\mu \\ &= \lim_{n\to\infty}\sum_{i=1}^m a_i\int_X \chi_{E_i } D\chi_{A_n}d\mu = \sum_{i=1}^m a_i\mu(E_i) =\int_X \varphi d\mu.
\end{align}
Hence qualities in \eqref{t} hold for each $f\in S\subset L^1(X)\cap L^{\infty}(X)$ and it means that
$$\int_X D^*f d\mu=\int_X f d\mu.$$
Therefore
$D:\mathcal{S}\to L^1(X)$ has a unique extension to a  doubly stochastic operator on $L^1(X)$. The reverse is easily verifiable.
\end{proof}
It is very obvious that with the same proof and very slight modification, we have the semi-doubly stochastic version of the above Theorem as follows.
\begin{theorem}\label{SDSprop*}
Let $S:\mathcal{S}\to L^1(X)$ be a linear function. Then $S$ has a unique extension to a semi-doubly stochastic operator if and only if $S$ is nonnegative and the following inequalities hold for each measurable set $E$ with $\mu(E)<\infty$:
\begin{equation}\label{EE1}
\int_X S\chi_E d\mu=\mu(E),
\end{equation}
\begin{equation}\label{EE2}
\lim_{n\to\infty}\int_X \chi_E S\chi_{A_n}d\mu\leq\mu(E).
\end{equation}
\end{theorem}
In this part we want to introduce another example of semi-doubly stochastic operators on $L^1(X)$ (Proposition \ref{Perfect Example}) which we will use it for characterization of majorization relation. From now on, as a contract unless otherwise stated, we will assume that $(X,\mathcal{A},\mu)$ is a $\sigma$-finite measure space, And $P:=\{ A_n :~~~ n\in \mathbb{N}\}$ is a disjoint family of measurable sets with $X=\bigcup_{n\in\mathbb{N}}A_n$ and such that $0<\mu(A_n)<+\infty$, for all $n\in \mathbb{N}$. Then it is easily seen that the map $\Phi_P:L^1(X)\to l^1$ given by 
\begin{align}\label{5}
 \Phi_P(f)=(\int_{A_1} f d\mu, \int_{A_2} f d\mu, \dots, \int_{A_n} f d\mu, \dots), \qquad  f\in L^1(X) 
\end{align}
is bounded linear map. Let 
${\Phi}^*_P: l^\infty\to L^\infty(X)$ be its adjoint.
Then for $(a_n)\in l^\infty$ and $f\in L^1(X)$,
$$\langle f, \Phi^*_P(a_n)\rangle=\langle \Phi_P (f), (a_n)\rangle=\sum_{n=1}^\infty a_n \int_{A_n} f d\mu= \int_X f \sum_{n=1}^\infty a_n \chi_{A_n}d\mu$$
Therefore, 
\[  \quad \Phi^*_P(a_n)=\sum_{n=1}^\infty a_n \chi_{A_n}, \qquad \forall (a_n)\in l^\infty.\]
Similarly, if $\Psi_P:l^1\to L^1(X)$ is defined by
\begin{align}\label{6}
\Psi_P(a_n)=\sum_{n=1}^\infty \frac{a_n }{\mu(A_n)}\chi_{A_n},\qquad \forall (a_n)\in l^\infty,
\end{align}
then $\Psi_P$ is also a bounded linear map with the adjoint $\Psi_P^*: L^\infty(X)\to l^\infty$  $\forall g\in L^\infty(X)$ is defined by 
\[  \Psi_P^*(g)=(\frac{1}{\mu(A_1)}\int_{A_1} f d\mu, \frac{1}{\mu(A_2)}\int_{A_2} f d\mu, \dots, \frac{1}{\mu(A_n)}\int_{A_n} f d\mu, \dots). \]
\begin{proposition}\label{Perfect Example}
The bounded linear map $G_P: L^1(X)\to L^1(X)$ defined by $G_P=\Psi_P\Phi_P$ is a doubly stochastic operator, and accordingly a semi-doubly stochastic.
\begin{proof}
Using the above considerations, we have 
$$\forall f\in L^1(X),\quad G_P(f)=\Psi_P(\Phi_P(f))=\sum_{n=1}^\infty (\frac{1 }{\mu(A_n)}\int_{A_n} f d\mu)\chi_{A_n},$$
and
\[\forall g\in L^\infty(X),\quad G^{*}_P(g)=\Psi_P^*(\Phi_P^*(g))=\sum_{n=1}^\infty (\frac{1 }{\mu(A_n)}\int_{A_n} g \,d\mu)\chi_{A_n}.\]
Clearly $G_P$ is a positive operator. For $f\in L^1(X)$, using monotone convergence Theorem, we have

\begin{align*}
\int_x G_P(f) d\mu & =  \sum_{n=1}^\infty (\frac{1 }{\mu(A_n)}\int_{A_n} f d\mu)\int_X \chi_{A_n} d\mu\\
& = \sum_{n=1}^\infty\int_{A_n} f d\mu = \int_X f d\mu.
\end{align*}
\\
If $g\in L^\infty(X)\cap L^1(X), $
\[\int_x \vert G^*_P(g)\vert d\mu \leq \sum_{n=1}^\infty (\frac{1 }{\mu(A_n)}\vert\int_{A_n} g d\mu\vert\int_X \chi_{A_n} d\mu\leq
\int_X \vert g\vert d\mu < \infty \]
i.e. $G^*_P(g)\in L^1(X)$. Similarly, 
$$\int_X G^*_P(g) d\mu = \sum_{n=1}^\infty \int_{A_n} g d\mu=\int_X g d\mu. $$
By definition $G_P\in \mathcal{D}(L^1(X))$, and accordingly $G_P\in S\mathcal{D}(L^1(X)).$
\end{proof}
\end{proposition}

The previous proposition is an spacial case of the following theorem when $D:=I$. To see why we will use extra assumption on measure space $X$ in general case (Theorem \ref{d}), first let $D:l^1\to l^1$ be a doubly stochastic operator and $\Phi_P:L^1(X)\to l^1$ and $\Psi_P: l^1\to L^1(X)$ are the maps defined in \eqref{5} and \eqref{6}, corresponding to the family of measurable subsets $A=\{A_n:~~~ n\in \mathbb{N}\}$ of $X$ with $X=\cup_{n\in \mathbb{N}} A_n$  and such that $0<\mu(A_n)<\infty,$ for all $n\in \mathbb{N}$. If $G_D:=\Psi_PD\Phi_P$ that for each $f\in L^1(X),$
$$G_D(f)= \sum_{n=1}^\infty (\frac{1 }{\mu(A_n)}(\sum_{j=1}^\infty d_{jn}\int_{A_j} f d\mu)\chi_{A_n}. $$
Hence 
\[\int_X G_D(f) d\mu = \sum_{n=1}^\infty \sum_{j=1}^\infty d_{jn}\int_{A_j} f d\mu=\sum_{j=1}^\infty\int_{A_j} f d\mu (\sum_{n=1}^\infty d_{jn})=\sum_{j=1}^\infty\int_{A_j} f d\mu\]\[ =\int_{X} f d\mu.  \]
 So to prove $G_D$ is a doubly stochastic operator (resp. semi- doubly stochastic), we should obtain the same equality relation (resp. inequality relation) for the map $G^*_D=\Psi^*_A D^*\Phi^*_A$.
\begin{theorem}\label{d}
Let ($X,\mathcal{A}, \mu$) be $\sigma$-finite or non-finite measure space and suppose \[\inf\{\mu(A_n), \quad n=1, 2 , \cdots\}=a\ne 0.\] 
Then,
\begin{itemize}
\item[(i)]
if $D:l^1\to l^1$ is semi-doubly stochastic operator, then the map $G_D:L^1(X)\to L^1(X)$ defined by $G_D=\Psi_P D \Phi_P$ is a semi-doubly stochastic operator.
\item[(ii)]
if $D:L^1(X)\to L^1(X)$ is a semi-doubly stochastic operator then the map $G_D:=\Phi_P D \Psi_P: l^1\to l^1$ is also semi-doubly stochastic on the sequence space $l^1$.
\end{itemize}
\end{theorem}
\begin{proof}
(i) For $g\in L^1(X)\cap L^\infty(X)$,
$$G^*_D(g)= \sum_{n=1}^\infty (\sum_{i=1}^\infty d_{in}\frac{1 }{\mu(A_i)}\int_{A_i} g d\mu)\chi_{A_n} $$
Therefore, using the assumption, we have 
\begin{align*}
\int_X G^*_D(g) d\mu &= \sum_{n=1}^\infty (\sum_{i=1}^\infty d_{in}\frac{1 }{\mu(A_i)}\int_{A_i} g d\mu)\mu(A_n)=\sum_{n=1}^\infty \sum_{i=1}^\infty d_{in}\int_{A_i} g d\mu \\
&= \sum_{i=1}^\infty \sum_{n=1}^\infty d_{in}\int_{A_i} g d\mu \leq \sum_{i=1}^\infty \int_{A_i} g d\mu=\int_X g d\mu.
\end{align*}
(ii) For each $n\in \mathbb{N}$, 
\[De_n= \Phi_P G \Psi_P (e_n)\]\[=\Phi_P (\frac{1}{\mu(A_n)} G (\chi_{A_n}))=(\frac{1}{\mu(A_n)}\int_{A_1} G (\chi_{A_n}) d\mu, \dots, \frac{1}{\mu(A_n)}\int_{A_n} G (\chi_{A_n}) d\mu, \dots)\]
Hence
\[\sum_{m=1} ^\infty \langle De_n , e_m\rangle=\frac{1}{\mu(A_n)}\sum_{m=1} ^\infty\int_{A_m} G (\chi_{A_n}) d\mu=\frac{1}{\mu(A_n)}\int_{X} G (\chi_{A_n}) d\mu\]\[=\frac{1}{\mu(A_n)}\int_{X} \chi_{A_n} d\mu=1. \]
Similarly,
\begin{align*}
\sum_{n=1} ^\infty \langle De_n , e_m\rangle &= \sum_{n=1} ^\infty \frac{1}{\mu(A_n)}\int_{A_m} G (\chi_{A_n}) d\mu =\sum_{n=1} ^\infty \frac{1}{\mu(A_n)}\int_{A_n} G^* (\chi_{A_m}) d\mu \\ &\le \frac{1}{a}\sum_{n=1} ^\infty\int_{A_n} G^* (\chi_{A_m}) d\mu = \frac{1}{a}\int_{X}   G^* (\chi_{A_m}) d\mu \\
&\leq\frac{1}{a}\int_{X} \chi_{A_m} d\mu = \frac{1}{a}\mu(A_m)=1.
\end{align*}
\end{proof}
\subsection{Equi-Integrability}
Bahrami and et.al characterized semi-doubly stochastic operators on $L^1(X)$ when $X$ is $\sigma$-finite  measure space by using the notion of majorization as follwoing theorem.
\begin{theorem}\cite[Theorem 2.4]{Semi}
Let $(X,\mu)$ be a $\sigma$-finite measure space and $S : L^1(X)\to L^1(X)$ be a positive bounded linear operator. Then for every non-negative integrable
function $f$ on $X$, $Sf\prec f$ if and only if $S\in S\mathcal{D}(L^1(X))$.
\end{theorem}
In this part we prove that when $(X,\mu)$ is a $\sigma$-finite measure space, $S_f=\{Sf:~~S\in S\mathcal{D}\} $ is equi-integrable. And then immediately give us relatively weakly compactness of $S_f$ when $(X,\mu)$ is a probability measure.
 \begin{definition}
Let $\mathcal{F}\subset L^1(X)$. $\mathcal{F}$ is said to be  equi-integrable if for every $\epsilon>0$ there exists some $\delta>0$ which for every $E\subset X$ with $\mu(E)<\delta$
\[\int_E|f|~d\mu<\epsilon, \qquad \forall f\in \mathcal{F}.\]
\end{definition}
For the proof of equi-integrability of $S_f=\{Sf:~~S\in S\mathcal{D}(L^1(X))\}$  when \linebreak $f\in L^1(X)$, first we prove the following lemma.
\begin{lemma}\label{eint}
Let $f,g\in L^1(X)$ be non negative functions with $f\prec g$ and also let $g\in L^\infty(X)$. Then for each $E\in X$ with $\mu(E)<\infty$
$$\int_E f~~d\mu\leq ||g||_\infty\mu(E)$$
\end{lemma}
\begin{proof}
By definition, for each $s>0$
$$d_g(s)=\mu(\{x\in X;~~g(x)>s \}).$$
Therefore for each $s\geq ||g||_\infty$, $d_g(s)=0$. On the other since $f\prec g$ for each $s>0$, 
$$\int_s^{\infty} d_f(\tau)d\tau\leq \int_s^\infty d_g(\tau)d\tau.$$
Hence for each $s\geq ||g||_\infty$, 
$$\int_s^\infty d_f(\tau)d\tau\leq \int_s^\infty d_g(\tau)d\tau=0$$
and then for each $s\geq ||g||_\infty$, $d_f(\tau)=0$. On the other hand 
\begin{align*}
    d_{f\chi_{E}}(s) &=\mu(\{x\in X,~~f(x)\chi_{E}(x)>s\})\\
    &=\mu(\{x\in E,~~f(x)>s\})\\ &\leq \min \{\mu(E), d_f(s)\},
\end{align*}
Therefore for each $s\geq ||g||_\infty$, $d_{f\chi_E}(s)=0$ and then 
$$\int_Ef d\mu=\int_Xf\chi_E d\mu=\int_0^\infty d_{f\chi_E}(s)ds\leq \mu(E)||g||_\infty.$$
\end{proof}
Now, we are ready for the following key theorem.
\begin{theorem}\label{equiint}
Let $f\in L^1(X)$ be non negative function. Then $S_f$ is equi-integrable. 
\end{theorem}
\begin{proof}
For $\epsilon>0$, non negative function $g\in L^1(X)\cap L^\infty(X)$ has been selected in such a way that $g\leq f$ and $||f-g||<\dfrac{\epsilon}{2}$. Now set $\delta = \dfrac{\epsilon}{2||g||_\infty}$, then for each measurable subset $E\subset X$ with $\mu(e)<\delta$ and each $S\in S\mathcal{D}(L^1(X))$, 
\begin{align*}
   \int_E|Sf|d\mu &\leq \int_E|Sf-Sg|d\mu+   \int_E|Sg|d\mu \\&\leq    \int_XS|f-g|d\mu+\int_E Sg d\mu \\
   &= ||f-g||+\int_E Sg d\mu < \dfrac{\epsilon}{2}+\int_E Sg d\mu,
\end{align*}
and since $Sg\prec g$ by using the Lemma \ref{eint}, 
$\int_ESgd\mu\leq ||g||_\infty\mu(E)$ and therefore
$$ \int_E |Sf|d\mu\leq \dfrac{\epsilon}{2} + ||g||_\infty\mu(E) <\dfrac{\epsilon}{2} + ||g||_\infty\delta<\epsilon. $$
\end{proof}
For probability space $(X,\mu)$ which has many applications in quantum sciences, the following theorem  provides lots of significant equivalence conditions for equi-integrability of $S_f=\{Sf:~~S\in S\mathcal{D}(L^1(X))\}$. 

\begin{theorem}\cite[Theorem 5.2.9]{book}\label{bo}  Let  $(X, \mu)$ is a probability measure space.
$\mathcal{F}$  be a bounded set in $L^{1}(X)$ .  Then the following conditions on $\mathcal{F}$ are equivalent.

({\it i}) $\mathcal{F}$ {\it is relatively weakly compact};

({\it ii}) $\mathcal{F}$ {\it is equi-integrable};

({\it iii}) $\mathcal{F}$ {\it does not contain a basic sequence equivalent to the canonical basis of}

$l^{1}$;

({\it iv}) $\mathcal{F}$ {\it does not contain a complemented basic sequence equivalent to the canonical basis of} $l^{1};$

({\it v}) {\it for every sequence} $(A_{n})_{n=1}^{\infty}$ {\it of disjoint measurable sets},
$$
\lim_{n\rightarrow\infty}\sup_{f\in F}\int_{A_{n}}|f|d\mu=0.
$$
\end{theorem}
Without loss of generality Theorem \ref{bo} holds under  the more general assumption that $(X,\mu)$ is finite measurable space.\\
Now we have the following corollary based on Theorem \ref{equiint}.

\begin{corollary}\label{Wcompact}
 Let $f\in L^1(X)$ be non-negative function where $(X,\mu)$ is finite measure space. Then $S_f$ is relatively weakly compact.
\end{corollary}
\begin{proof}
For each $S\in S\mathcal{D}(L^1(X))$ and a fixed 
 $f\in L^1(X)$ by definition of $S\mathcal{D}(L^1(X))$ we have
$$\int_X Sf ~d\mu=\int_X f ~d\mu.$$ 
Then for each $S\in S\mathcal{D}(L^1(X))$ and a fixed 
 $f\in L^1(X)$ we have $||Sf||_1=||f||_1$ so $S_f$ is bounded. Therefore by using Theorem \ref{equiint} and equivalency of items $(i)$ and $(ii)$ in Theorem \ref{bo},  $S_f$ is relatively weakly compact.

\end{proof}
\section{ Characterization of majorization on $L^1(X)$}\label{Sec4}
The goal of this section is to give a full characterization of majorization, the answer to Mirsky's question, using semi-doubly stochastic operators.

We already recalled Theorem \ref{main} as an extension of 
Theorem \ref{majvect}, $(1)\iff (2)$. In this section we want to consider the relation between Sublinear functions and also convex functions with majorization on $L^1(X)$  which gain us the extension of Theorem \ref{majvect}, $(2)\iff (3)$. Also we will provide a strong relation between integral operators and semi-doubly stochastic operators, and finally a fully characterization of majorization on $L^1(X)$ when $X$ is $\sigma$-finite measure space.

\subsection{Sublinear and Convex functions}
In matrix space, Dahl proved an equivalent condition for matrix majorization using sublinear functionals (i.e convex and positively homogeneous maps). Moein et.al proved one side extension of Dahl's result as the following theorem.
\begin{theorem}\cite[Theorem 3.8]{SRS}
If $(X, \mu)$ is $\sigma$-finite measure space and $f,g\in L^1(X)$ such that $f$
is matrix majorized by $g$(i.e there exists a Markov operator $M$ that $f=Mg$), then
  \[
        \int_X\varphi(f) d\mu\leq \int_X \varphi(g) d\mu, \]
for all sublinear functionals $\varphi:\mathbb{R}\to\mathbb{R}^+$.
\end{theorem}
Since $S\mathcal{D}(L^1(X))\subset \mathcal{M}(L^1(X))$ we simply have the following corollary.
\begin{corollary}\label{sublinear}
Let $(X,\mu)$ be a $\sigma$-finite measure space, $f,g\in L^1(X,\mu)$ and there exists $S\in S\mathcal{D}(L^1(X))$ with $f=Sg$. Then for all sublinear functionals $\varphi:\mathbb{R}\to \mathbb{R}^+$,
    \begin{equation}\label{subeq}
        \int_X\varphi(f) d\mu\leq \int_X \varphi(g) d\mu. \end{equation}
\end{corollary}
But for converse of the above corollary we have a counterexample.
\begin{example}
Let $X=[0,\infty)$, $\mu$ be Lebesgue measure and $f=3\chi_{[0,1)}+\cfrac{1}{2}\chi_{[1,2)}$ and $g=2\chi_{[0,2)}$. Then neither $f\prec_w g$ nor $g\prec_w f$. Now let $\varphi:\mathbb{R}\to \mathbb{R}^+$ be an arbitrary sublinear function then by its convexity and positively homogeneous property we have 
$$\varphi(f)=3\varphi(\chi_{[0,1)})+\cfrac{1}{2}\varphi(\chi_{[1,2)}),~ \text{and}~~\varphi(f)=2\varphi(\chi_{[0,2)}).$$
And then by linearity of integral we have
$$
\int_X\varphi(f)d\mu=3\int_{0}^{1}\varphi(1)d\mu+\cfrac{1}{2}\int_{1}^{2}\varphi(1)d\mu\leq 2\int_{0}^{2}\varphi(1)=\int_X\varphi(g)d\mu.
$$

\end{example}
Because each sublinear function is convex a natural question can arise, is the converse of Corollary \ref{sublinear} true for convex functions? or can we characterize majorization relation in $\sigma$-finite measure space with inequality in Corollary \ref{sublinear} based on convex functions? Chong answered this question positively for weak majorization but with a restriction for convex functions, convex functions have to be increasing.
\begin{theorem}\cite[Theorem 2.1]{Chong}\label{Chong0}
Let $(X,\mu)$ be infinite measure space and $f,g\in L^1(X)$ be non-negative. Then $f\prec_w g$ i.e
\[
\int_0^s f^\downarrow dm \leq  \int_0^s g^\downarrow dm,\quad \forall\, 0\leq  s\leq  \infty\]
if and only if for all non-negative increasing convex functions $\phi:\mathbb{R}^{+}\to\mathbb{R}^{+}$ with $\phi(0)=0$, 
\[\int_X \phi(f)d\mu\leq \int_X \phi(g)d\mu.\]
\end{theorem}

\subsection{Semi-doubly Stochastic and Integral Operators}\label{IntegralOp}

Stochastic (respectively doubly stochastic) integral operators as special classes of linear operators which are defined as follows are Markov (respectively doubly stochastic) operators.

\begin{definition}\cite{SRS}
A measurable function $K:X\times Y\to[0,\infty)$ is called \emph{stochastic kernel} if $\int_X K(x,y)d\mu(x)=1$ for almost all $y\in Y$, and is called a \emph{doubly stochastic kernel} if  stochastic  kernel has the additional property that $\int_Y K(x,y)d\nu(y)=1$ for almost all $x\in X$.
\end{definition}

\begin{definition}\cite{SRS} An integral operator $A:L^1(Y)\to L^1(X)$ defined by $Ag=\int_{Y}K(x,y)g(y)d\nu(y)$ is said to be a \emph{stochastic integral operator} (resp.\ \emph{doubly stochastic integral operator}) if $K(x,y)$ is stochastic kernel (resp.\ doubly stochastic kernel).
\end{definition}

Each stochastic integral operator is a Markov operator (stochastic operator), and also each doubly stochastic integral operator is a doubly stochastic operator. But, simply by considering the identity operator which is a doubly stochastic operator, it is clear that the converse of both statements is false. 
In spite of that in \cite{SRS} is proven that a  Markov (resp.\ doubly stochastic) operator $D:L^1(Y)\to L^1(X)$ on a finite dimensional subspace $F$ of $L^1(Y)$ can be approximated by stochastic (resp.\ doubly stochastic) integral operators when $(X,\mu )$ and $(Y,\nu)$ be a $\sigma$-finite (resp.\ finite) measure spaces.
\begin{theorem}\label{finite aprox}
If $(X,\mu )$ and $(Y,\nu)$ are finite measure spaces, then $D$ as a doubly stochastic operator from $L^1(Y)$ to $L^1(X)$ on a finite subspace $F$ of $L^1(Y)$ can be approximated by doubly stochastic integral operators.
\end{theorem}
\begin{proof}
See Theorem $3.7$ in \cite{SRS}.
\end{proof}
By using this fact from \cite[Proposition 2.6]{Semi} that for the finite measure space $(X,\mu)$ the set of semi-doubly stochastic operators and the set of doubly stochastic operator coincide, a natural question which is our first aim in this section arises. Can we extend the Theorem \ref{finite aprox} to $\sigma$-finite measure space?
 For this purpose, we need a class of integral operators between stochastic integral operators and doubly stochastic integral operators.
\begin{definition}
A measurable functional $K:X\times Y\to[0,\infty)$ is called \emph{semi-doubly stochastic kernel} if $\int_X K(x,y)d\mu(x)=1$ for almost all $y\in Y$, and \linebreak $\int_Y K(x,y)d\nu(y)\leq 1$ for almost all $x\in X$.
\end{definition}
\begin{definition} An integral operator $A:L^1(Y)\to L^1(X)$ defined by $Ag=\int_{Y}K(x,y)g(y)d\nu(y)$ is said to be a \emph{semi-doubly stochastic integral operator} if $K(x,y)$ is semi-doubly stochastic kernel.
\end{definition}
\begin{lemma}\cite[Lemma 3.5]{SRS}\label{lemma3.5}  Let $(X,\mu)$ be a $\sigma$-finite measure space and let $F$ be a finite dimensional subspace of $L^1(X)$.  Then there exists a sequence of partitions $\{P_n\}_{n=1}^\infty$ of $X$  into disjoint sets of finite measure such that $\{{G_P}_nf\}_{n=1}^\infty$ converges to $f$ in the $L^1$ norm for all $f\in F$. \end{lemma}
\begin{theorem}\label{integral op}
Let $(X,\mu )$ 
be $\sigma$-finite measure space, and $S:L^1(X)\to L^1(X)$ be a semi-doubly stochastic operator. Then there exists a sequence of semi-doubly stochastic integral operators on $L^{1}(X)$ which converge to $S$ on a finite dimensional subspace $F$ of $L^1(X)$.
\end{theorem}
 \begin{proof} (The proof is combination the results of \cite{Semi} and \cite{SRS} but after modification in terms of semi-doubly stochastic). From the Proposition \ref{Perfect Example},
 we have 
\[ G_P(f)=\Psi_P(\Phi_P(f))=\sum_{n=1}^\infty (\frac{1 }{\mu(A_n)}\int_{A_n} f d\mu)\chi_{A_n},\qquad \forall f\in L^1(X),\]
is a  semi-doubly stochastic operator on $L^{1}(X)$. Since the composite of two  semi-doubly stochastic operators is semi-doubly stochastic operator, $G_PS$ is a semi-doubly stochastic operator on $L^{1}(X)$, we will show that it is a semi-doubly stochastic integral operator.  Since $F$ is a finite dimensional subspace of $L^{1}(X)$, then the forward image $S(F)$ is a finite dimensional subspace of $L^{1}(X)$ as well.  And then the result follows from Lemma \ref{lemma3.5}.\\ Fix $x\in X$, then there exists a unique $A_k\in P$ such that $x\in A_k$. The boundedness follows from  
$$\begin{array}{rcl}  \vert (G_PS(f))(x)\vert &=& \vert \frac{1}{\mu(A_k)}\int_{A_k}(Sf(t))d\mu(t)\vert \\  &\le &  \frac{1}{\mu(A_k)}\int_{A_k}\vert (Sf)(t)\vert d\mu(t) \\ & \le&  \frac{1}{\mu(A_k)}\int_{X}\vert (Sf)(t)\vert d\mu(t) \\ & \le & \frac{1}{\mu(A_k)}\int_{X}\vert f(x)\vert d\mu(x)\\ & = &\frac{1}{\mu(A_k)}\Vert f\Vert_1.\end{array}$$  Hence by using the Riesz representation theorem, there exist $h_x\in L^{\infty}(X)$ which is a nonnegative function and $(G_PS(f))(x)=\int_X f(y)h_x(y)d\mu$.  Now let $K_{P}(x,y)=h_x(y)$ for all $x,y\in X$.  Then \begin{equation}\label{m1}
     (G_PSf)(x)=\int_X K_{P}(x,y)f(y)d\mu(y)
 \end{equation}  

Since $G_PS$ is a semi-doubly stochastic operator then it preserves the  integral 
$$\int_X ((G_PS)f) d\mu=\int_X fd\mu\quad \forall f\in L^1(Y)$$ and by using Fubini's Theorem
$$\begin{array}{rcl}
\int_X f(x)d\mu(x) &=& \int_X ((G_PS)f)(x) d\mu(x)\\ &=& \int_X\int_X K_{P}(x,y)f(y)d\mu(y)d\mu(x)\\
&=&\int_X f(y) (\int_X K_{P}(x,y)d\mu(x))d\mu(y).
\end{array}$$

Therefore $\int_X K_{P}(x,y)d\mu(x)=1$ for almost all $y\in X$.  
   
  Now suppose $\{B_{n};n\in \mathbb{N}\}$ is an increasing sequence of measurable sets such that  $X=\displaystyle \bigcup_{n\in \mathbb{N}}B_{n}$ and for each $n\in \mathbb{N}$, $\mu(B_{n})<\infty$.
But because of equality of dual pairing as follows
$$ \langle G_PS\chi_{B_{n}},\chi_{A} \rangle = \langle \chi_{B_{n}},((G_PS)^{*}\chi_{A}) \rangle, $$
we have $$
\int_{X}(G_PS\chi_{B_{n}})\chi_{A}\mathrm{d}\mu(x)=\int_{X}\chi_{B_{n}}((G_PS)^{*}\chi_{A})\mathrm{d}\mu(x)\rightarrow\int_{X}(G_PS)^{*}\chi_{A}\mathrm{d}\mu(x).
$$
And by using equation \eqref{m1} for $f=\chi_{B_{n}}$ for each $n$ we have 
\[\int_{X}(G_PS\chi_{B_{n}})\chi_{A}\mathrm{d}\mu(x)=\int_{X}\int_{X} K_P(x,y)\chi_{B_{n}}d\mu(y)\chi_{A} d\mu(x)\]\[=\mu(A)\int_{B_{n}}K_P(x,y)\mathrm{d}\mu(y).\]
 Then if $n\to\infty$ we have 
  \[ \int_{X}(G_PS)^{*}\chi_{A}\mathrm{d}\mu(y)=\mu(A)\int_X K_P(x,y)\mathrm{d}\mu(y).\]
 Since $G_PS$ is semi-doubly stochastic then according to its definition for every measurable subset $A$ with finite measure we have 
 \begin{equation*}\label{E02}
     \displaystyle \int_{X}(G_PS)^{*}\chi_{A}\mathrm{d}\mu\leq\mu(A).  \end{equation*} Then
 $$\mu(A)\int_X K_P(x,y)\mathrm{d}\mu(y)\leq \mu(A),$$
 therefore $\int_X K_{P}(x,y)d\mu(y)\leq 1$ for almost all $x\in X$ and hence $G_PS$ is a semi-doubly stochastic integral operator.
\end{proof}

The following theorem is a combination of two well-known results in finite measurable space by Chong and Day. 
\begin{theorem}\label{Chong-Day}
Let $f,g\in L^1(X,\mu)$ and $\mu(X)<\infty$. Then the following are equivalent:
\begin{itemize}
    \item[a.]
    $f\prec g.$
    \item[b.]
    For all convex functions $\varphi:\mathbb{R}\to \mathbb{R}$,
    $$\int_X\varphi(f) d\mu\leq \int_Y \varphi(g) d\mu.$$
    \item[c.] There exists a doubly stochastic operator $D$ on $L^1(X)$ such that $f=Dg$.
\end{itemize}
\end{theorem}
\begin{proof}
The equivalence of $(a)$ and $(b)$ is proved by Chong \cite[Theorem 2.5]{Chong} and the equivalence of $(a)$ and $(c)$ is proved by Day \cite[Theorem 4.9]{Day}.
\end{proof}
Now we can summarize all equivalent conditions for majorization relation in case of $\sigma$-finite measure space. Theorem \ref{main2} as an extension of Day-Chong's result i.e  Theorem \ref{Chong-Day} or Hardy, Littlewood and P{\'o}lya's result to $\sigma$-finite measure space asserts strong relation between majorization based on decreasing rearrangement, semi-doubly stochastic operators, semi-doubly stochastic integral operators and convex functions integral inequality. Also we should mention, there is an open problem in \cite{MM} (the converse of corollary II.13) which the following theorem can solve it by using counting measure on $\mathbb{N}$.

\begin{theorem}\label{main2}
Let $(X,\mu)$ be $\sigma$-finite measure space and non-negative $f,g\in L^1(X,\mu)$. Then the followings are equivalent:
\begin{itemize}
\item[1.] $g\prec f$.
    \item[2.] There exist a sequence $(S_k)_{n=1}^\infty$ of semi-doubly stochastic operators on $L^1$ such that $S_kf\to g$ in $L^1(X)$.
    \item[3.] There exist a sequence $(I_n)_{n=1}^\infty$ of semi-doubly stochastic integral operators on $L^1$ such that $I_nf\to g$ in $L^1(X)$.
    \item[4.] $\int_X \phi(g)d\mu\leq \int_X \phi(f)d\mu$ for all increasing convex function $\phi$ from $\mathbb{R}^+$ to $\mathbb{R}^+$ such that $\phi(0)=0$, and $\int_X g d\mu=\int_X f d\mu$.
    \item[5.] $\int_X (g-u)^+ d\mu\leq \int_X (f-u)^+d\mu$ for each positive real number $u$, and $\int_X g d\mu=\int_X f d\mu$.
\end{itemize} 
\end{theorem}
\begin{proof}
$(1)$, $(2)$ and $(4)$ are equivalent by Theorems \ref{main} and \ref{Chong0}.\\ $(2)$ implies $(3)$: Since there exists a sequence of semi-doubly stochastic operators $\{S_k\}_{k=1}^\infty$ which $S_kg\to f$, in $L^1(X)$ by using Theorem \ref{integral op}, for each $k\in\mathbb{N}$, $S_k$ as a semi-doubly stochastic operator  on a finite dimensional subspace $F=\mbox{span}\{g\}$ of $L^1(X)$ can be approximated by semi-doubly stochastic integral operators means that for each $k\in\mathbb{N}$ there exists a sequence of semi-doubly stochastic integral operators $\{I_n\}_{n=1}^\infty$ on $L^{1}(X)$ which converge to $S_k$ on $F$. 
Then we can say there exists a sequence of semi-doubly stochastic integral operators $I_n$ which $I_ng\to f$ in $L^1(x)$.\\
$(3)$ implies $(2)$: we show that each semi-doubly stochastic integral operator is a semi-doubly stochastic operator.\\
Let $A:L^1(X)\to L^1(X)$ be a semi-doubly stochastic integral operator. Then by definition for each $f\in L^1(X)$ defined by a semi-doubly kernel $K(x,y)$ as follows
 $$Ag=\int_{X}K(x,y)g(y)d\mu(y)$$ 
we have to show that $A\in S\mathcal{D}(L^1(X))$ i.e.  for all $f\in L^1(X)$ and for each $E\in \mathcal{A}$ with $\mu(E)<\infty$, 
\begin{equation}\label{EEE1}
  \int_X Af ~d\mu=\int_X f ~d\mu,
\end{equation}
and also
\begin{equation}\label{EEE2}
     \int_X A^*\chi_E ~d\mu\leq \mu(E).
\end{equation}
\eqref{EEE1} is simply obtained from this fact that $\int_X K(x,y)d\mu(x)=1$ for almost all $y\in X$.   
Also \eqref{EEE2} is obtained from this fact that for every measurable subset $E$ with finite measure 
\[\lim_{n\to\infty}\int_{X}(A\chi_{B_{n}})\chi_{E}\mathrm{d}\mu(x)=\int_{X}\chi_{B_{n}}(A^{*}\chi_{E})\mathrm{d}\mu(x)=\int_{X}A^{*}\chi_{E}\mathrm{d}\mu(x),\]
 and since
\[A\chi_{B_{n}}=\int_X K(x,y)\chi_{B_{n}}d\mu(y)\]
by using this fact that $\int_Y K(x,y)d\nu(y)\leq 1$ for almost all $x\in X$, when $n\to \infty$ we obtain \eqref{EEE2}.\\
$(4)$ implies $(5)$ simply by choosing convex function $\phi$ given by $\phi(g)=\max\{g-u,0\}=(g-u)^+$ for each $g\in L^1(X)$. \\ Finally by using Theorems \ref{Chong} and \ref{Chong2}, $(5)$ implies $(1)$.
\end{proof}
\begin{remark}
As you can see in Definition \ref{DecRe}, non-negative property for functions in $L^1(X)$ is a necessary condition for defining majorization based on decreasing rearrangement. But Theorem \ref{main2} allows us to extend the definition of majorization to arbitrary functions which are not necessarily non-negative. 
\end{remark}

\section*{Acknowledgments}
This work was supported by the Department of Mathematical Sciences at the Isfahan University of Technology, Iran.

\bibliographystyle{unsrt}  

\end{document}